\documentclass[conference,doublecolumn]{IEEEtran}
\usepackage{amssymb}
\usepackage{graphicx}
\usepackage{graphics}
\usepackage{mathrsfs}
\usepackage{amsmath}
\usepackage{amsthm}
\usepackage{amsfonts}
\usepackage{cite}

\def \hh{\vskip 0.5\baselineskip \hbox to \hsize}

\newtheorem{definition}{Definition}
\newtheorem{theorem}{Theorem}
\newtheorem{lemma}{Lemma}

\begin{document}

\title{The Impact of Incomplete Information on Games in Parallel Relay
Networks}
\author{{\large Hongda Xiao and Edmund M. Yeh} \\
{\normalsize Department of Electrical Engineering, Yale University, USA}}
\maketitle

\begin{abstract}
We consider the impact of incomplete information on incentives
for node cooperation in parallel relay networks with one source node, one
destination node, and
multiple relay nodes. All nodes are selfish and strategic, interested in
maximizing their own profit instead of the social welfare. We consider the
practical situation where the channel state on any given relay path is
not observable to the source or to the other relays.  We examine different bargaining
relationships between the source and the relays, and propose a framework for
analyzing the efficiency loss induced by incomplete information. We analyze
the source of the efficiency loss, and quantify the amount of inefficiency
which results.
\end{abstract}




\section{{\protect\large {Introduction}}}

There is now widespread awareness of the importance of incentives in the
management of communication networks \cite{paper:BLV05, paper:CGKO04, paper:MQ05, book:BH07, paper:KMT98, paper:CDR03}. Network nodes often cannot be relied
upon to cooperatively implement network algorithms in the service of the
social good. Instead, selfish nodes will behave in a given manner only if it
is profitable for them to do so. Of clear interest is the impact of such
selfish actions on the social good. From the network point of view, it is
important to design incentives such as pricing schemes, which induce selfish
behavior aligned with the social good.

In single-hop networks, the incentive issue and its impact on social
efficiency have been extensively studied. In \cite{Roughgarden02,paper:Rou05}, the
authors considered the Nash Equilibrium for selfish routing, in which source
packets choose paths to the destination to minimize their individual
latency, rather than complying with a global routing algorithm to achieve
social optimality. In \cite{Walrand05} and \cite{Srikant05}, the authors
consider network service pricing for internet service providers. They showed
that cooperation among multiple service providers is required when their
links are used by common users. In \cite{Ozdaglar08}, the authors study
competitive behavior among multiple parallel links, and characterized the
efficiency loss due to competition.

The issue of incentives has also been investigated for multi-hop networks. A
number of papers~\cite{Neely07, Yang03, selfishadhoc} advocate the use of
credits to provide incentives for network nodes to cooperate. In \cite%
{Roughgarden08}, the authors investigate the impact of heterogeneous traffic
on the pricing of network service providers. Selfish behavior has also been
investigated in the context of cooperative relay networks. In \cite{Yufang06}%
, the authors considered a nonlinear pricing game, where the relay nodes
propose nonlinear charging functions to the source, and the source allocates
the traffic to minimize the payment to relay nodes. In \cite{Liu09}, the
authors considered a Stackelberg bargaining game, in which the relay nodes
cooperate as one party in competing with the source node.

All the above papers assume a complete information setting where players in
the network game have complete knowledge about quantities such as the state
of network links. In practice, this assumption is often too strong.
Information regarding network quantities is typically incomplete and
imperfect. In an internet service provider (ISP) pricing game, for instance,
the characteristics and service requirements of the users can be opaque to
the service providers \cite{Srikant07}. In a multi-hop network such as the
Internet, a source does not typically have perfect information on the
congestion state of links a few hops away \cite{Gallager}. Finally, in
wireless networks, the source usually cannot observe or test the channel
state from a relay to the destination. Neither can a relay observe the
channel state from other relays to the destination. Given the above, it is
clear that in analyzing selfish behavior in network settings, the role of
incomplete information must be emphasized.

One approach to network design problems with incomplete information is
through dominant implementable mechanisms \cite{Nisan99}. This idea has been
used in the context of spectrum auctions \cite{Archer04} and communication
networks \cite{Johari05}. These mechanisms, however, require a centralized
authority and extra funding from an outsider. This makes the extension to
general multi-hop networks difficult. Another approach, based on the idea of
Bayesian Nash Equilibrium, a generalization of the Nash Equilibrium concept,
is advocated in \cite{Gairing08}. Here, the authors consider selfish routing
in a single-hop network, where every source node knows only its own traffic
requirement, but has knowledge of the traffic distribution of other sources.
While the results in \cite{Gairing08} are appealing, it remains unclear how
they might extend to the multi-hop network situation.

In this work, we investigate the impact of incomplete information on the
problem of pricing and incentives in a two-hop parallel relay network.
We consider two scenarios, one in which the source has limited bargaining power
and one in which the source has full bargaining power.
In the limited bargaining power scenario, the source can only react passively to the relays'
signals, and the game can be considered to be a pricing game.
For this case, we show that all Nash Equilibria in the complete information
game are efficient, including those induced by linear charging functions. We
then characterize the Bayesian Nash Equilibrium for the incomplete
information game in which relays propose linear pricing functions, and
show that incomplete information can induce inefficiencies, which are
exacerbated by asymmetric prior knowledge on the type distribution.   Next,
in the scenario where the source has full bargaining power, the source
is allowed to provide a general contract.  For this case, we first show that in the
game with complete information, (Bayesian) Nash equilibria exist and are all
efficient. Next, we investigate the game with incomplete information. To
deal with the difficulty of characterizing the Bayesian Nash Equilibria in
this case, we first show that if a resource allocation outcome can be
realized by a Bayesian Nash equilibrium, then there exists a ``truth
telling" Bayesian Nash equilibrium that realizes the outcome. We then show
that the set of outcomes for the ``truth telling" Bayesian Nash equilibria is
included in the set of outcomes for
the Nash equilibria for a {complete information game}, in
which the link cost functions are replaced by specified ``virtual cost
functions." Using this approach, we obtain for a symmetric network
scenario a bound on the amount of inefficiency which may result from
incomplete information.

\section{{\protect\large {Network Model}}}
\label{model}
\subsection{{\protect\large \textsl{Network Traffic Allocation}}}

In wireline and wireless networks, it is often the case that an information
source cannot directly reach its destination, but must do so with the aid of
intermediate relays. We model such a situation as follows. Consider a
parallel relay network modelled by a directed graph $G=(V,E)$, with a single
source $s$, destination $d$, and a set of relays $I$, where $|I|=n$. We
assume that there is no direct link between $s$ and $d$. Instead, The relays
in $I$ are used to forward traffic in a two-hop fashion from $s$ to $d$.



The source wishes to maintain a certain rate of transmission with the
destination. We shall consider two scenarios. In the first \emph{inelastic}
scenario, the source has a fixed rate $r_s$ of transmission. This rate must
be carried by the relays in $I$, where the traffic rate forwarded by relay $%
i $ is $r_i$, and $\sum_{i =1}^n r_i = r_s$. In the second \emph{elastic}
scenario, the source may be willing to withhold some of its transmission
rate, according to how the cost of sending traffic affects it overall
utility. Let $r_0$ denote the amount of traffic withheld or rejected. Then $%
r_s - r_0$ is the total admitted traffic from the source. A traffic vector $%
\mathbf{r} \triangleq (r_{0},r_1, \ldots, r_n) \in {\mathbb{R}}_+^{n+1}$ is
a feasible routing of the source traffic if it satisfies $r_0 + \sum_{i=1}^n
r_i = r_s$.

\subsection{{\protect\large \textsl{Cost Function and Utility Function}}}

\label{sec:model} In general, for any relay node $i$, there is a cost
involved in forwarding traffic for source $s$. This cost typically depends
both on the properties of the links adjacent on relay $i$ and the amount of
traffic flowing through those links. Denote the traffic flow on link $(i,j)
\in E$ by $f_{ij}$. We assume that link $(i,j)$ has a cost function $%
C_{ij}(\theta _{ij},f_{ij})$ with $C_{ij}(\theta _{ij},0)=0$, where $\theta
_{ij}$ is a measure of the quality of link $(i,j)$. This quality may have
different physical meanings in different contexts. For example, if the cost
function reflects the queuing delay on $(i,j)$, then using the M/M/1
approximation, $C_{ij}(\theta _{ij},f_{ij})=\frac{f_{ij}}{k_{ij}-f_{ij}}$.
Here, $\theta_{ij}$ denotes the link capacity $k_{ij}$. For another example,
consider the cost of power assumption required for transmitting traffic of
rate $f_{ij}$ over a wireless link with channel gain $g_{ij}$, bandwidth $W$%
, and receiver noise power $N$. Using the Shannon capacity formula, we have $%
f_{ij} = W\log (1+g_{ij}P_{ij}/N)$, where $P_{ij}$ is transmission power
required on link $(i,j)$. Thus, the link cost is
\begin{equation*}
C_{ij}(\theta_{ij},f_{ij}) =\frac{N}{g_{ij}}(2^{f_{ij}/W}-1).
\end{equation*}
Here, $\theta _{ij}$ denotes the channel gain $g_{ij}$.

\begin{figure}[tbp]
\centering
\includegraphics[scale = 0.45, bb=194 212 528 497]{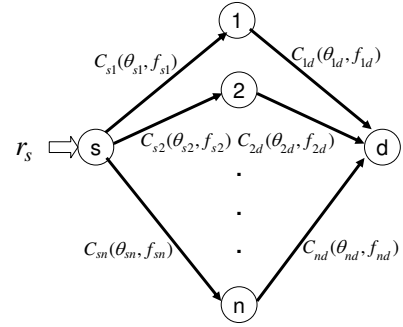}
\caption{relay channel}
\end{figure}

Now consider the overall cost $C_{i}(\theta _{i},r_{i})$ for relay node $i$
to forward traffic of rate $r_{i}$ from source $s$ to destination $d$, where
$\theta _{i}$ measures the quality or \emph{type} of the \emph{path} from $s$
to $d$ through $i$. We assume that $C_{i}(\theta _{i},r_{i})=C_{si}(\theta
_{si},r_{i})+C_{id}(\theta _{id},r_{i})$. The costs $C_{i}(\theta
_{i},r_{i}) $ are particularly amenable to analysis if $\theta _{i}$ can be
expressed as a simple scalar function of $\theta _{si}$ and $\theta _{id}$: $%
\theta _{i}=h(\theta _{si},\theta _{id})$. This is true in the example of
the power consumption cost function given above, where $\theta _{ij}=g_{ij}$
is the channel gain on link $(i,j)$. Normalizing the bandwidth and receiver
noise power to 1, we have
\begin{eqnarray}
C_{i}(\theta _{i},r_{i}) &=&P_{si}+P_{id}  \notag \\
&=&(2^{r_{i}}-1)/g_{si}+(2^{r_{i}}-1)/g_{id}  \notag \\
&=&(2^{r_{i}}-1)\theta _{i}^{-1},
\label{Cost Function}
\end{eqnarray}%
where $\theta _{i}\triangleq (g_{si}^{-1}+g_{id}^{-1})^{-1}=(\theta
_{si}^{-1}+\theta _{id}^{-1})^{-1}$. In this paper, we focus on situations
where the path quality $\theta _{i}$ can be expressed as a scalar function
of $\theta _{si}$ and $\theta _{id}$. We further assume that $\theta _{i}$
belongs to a compact interval $[\underline{\theta }_{i},\overline{\theta }%
_{i}]$.

Motivated by the power consumption example, we assume that $%
C_{i}(\theta_{i},r_{i})$ is twice continuously differentiable on $[%
\underline{\theta}_i, \overline{\theta}_i] \times [0, r_s]$, and strictly
increasing and convex in $r_i$: ${\partial {C_{i}(\theta _{i},r_{i})}}/{%
\partial {r_{i}}} > 0$ and ${\partial ^{2}{C_{i}(\theta _{i},r_{i})}}/{%
\partial {r_{i}^{2}}} > 0$. Also, assume that $C_{i}(\theta_{i},r_{i})$ is
strictly decreasing in $\theta_i$: ${\partial {C_{i}(\theta _{i},r}_{i}{)}}/{%
\partial {\theta _{i}}} < 0$. Furthermore, assume ${\partial ^{2}{C_i(\theta
_{i},r_{i})}}/{\partial {\theta _{i}}\partial{r_{i}}}\leq 0$.

Now consider the source $s$. In the inelastic case, source $s$ sends traffic
at a fixed rate $r_{s}$ into the network. In the elastic case, source $s$
may withhold traffic of rate $r_{0}$ from the network, and send the other
part of the traffic $r_{s}-r_{0}$ into the network. Let the utility function
of the source be given by $W_{s}(\theta _{s},r)$, where $\theta _{s}\in
\lbrack \underline{\theta }_{s},\overline{\theta }_{s}]$ parameterizes the
utility for the source, and $r$ is the source rate admitted into the
network. For example, the source utility may be $W_{s}(\theta _{s},r)=\theta
_{s}\log (1+r)$. Assume that $W_{s}(r)=W_{s}(r_{s})$ for all $r\geq r_{s}$,
i.e. $r_{s}$ is the maximum desired source rate. $W_{s}(\theta _{s},r)$ is
assumed to be continuously differentiable, strictly increasing and concave
in $r$ on $[0,r_{s}]$. Let $C_{s}(\theta _{s},r_{0})\triangleq
W_{s}(r_{s})-W_{s}(r_{s}-r_{0})$ denote the source's \emph{utility loss}
from having traffic of rate $r_{0}$ withheld from the network. Equivalently,
if $r_{0}$ is regarded as the traffic rate routed on a \emph{virtual
overflow link} directly from $s$ to $d$~\cite{Gallager}, then $C_{s}(\theta
_{s},r_{0})$ represents the cost on the overflow link when the link
parameter is $\theta _{s}$ and the flow rate is $r_{0}$. Since $W_{s}(r_{s})$
is a constant, it can be seen that $C_{s}(\theta _{s},r_{0})$ is
continuously differentiable on $[\underline{\theta }_{s},\overline{\theta }%
_{s}]\times \lbrack 0,r_{s}]$, strictly increasing and convex in $r_{0}$: ${%
\partial {C}_{s}{(\theta _{s},r_{0})}}/{\partial r{_{0}}}>0$ and ${\partial
^{2}{C}_{s}{(\theta _{s},r_{0})}}/{\partial r{_{0}^{2}}}>0$. Furthermore, we
assume that $C_{s}(\theta _{s},r_{0})$ is strictly decreasing in $\theta
_{s} $: ${\partial {C}_{s}{(\theta _{s},r_{0})}}/{\partial \theta {_{s}}}<0$%
. Finally, it can be seen that $C_{s}(\theta _{s},0)=0$ for all $\theta _{s}$%
. It can easily be checked that these properties are satisfied for the
example $W_{s}(\theta _{s},r)=\theta _{s}\log (1+r)$, for which $%
C_{s}(\theta _{s},r_{0})=W_{s}(r_{s})-\theta _{s}\log (1+r_{s}-r_{0})$. With
the aid of the virtual overflow link, we may view a game with an elastic
source as a game with an inelastic source of rate $r_{s}$ and an overflow
link $(s,w)$ with cost function $C_{s}(\theta _{s},r_{0})$.

\subsection{{\protect\large \textsl{Socially Optimal Allocation}}}

\label{sec:efficient}

A socially optimal traffic allocation in a parallel relay network is an
allocation which minimizes the total network cost, assumed to be the sum of
the link costs. Such an allocation can be realized through cooperation of
the network nodes. In networks with selfish and strategic nodes, a socially
optimal allocation may or may not be realizable. Nevertheless, the optimal
allocation serves as an important benchmark with which to measure the amount
of potential inefficiency introduced by selfish and strategic behavior.

Let $R \triangleq \{(r_{0},r_{1},...,r_{n})$$: r_j \geq 0~\forall j=0, \ldots,
n, \sum_{j=0}^n r_j = r_s\}$ be the set of feasible traffic allocations, and
let $\mathbf{r} \in R$ denote the vector of traffic rates in the network,
where $r_0$ is the rate withheld by the source, and $r_i$ is the rate routed
to relay $i, i=1, \ldots, n$. Note that for the case of an inelastic source,
$r_0 = 0$.

\begin{definition}
A traffic allocation vector $\mathbf{r}^*$ is called socially optimal if
\begin{equation}
\mathbf{r}^* \in \arg \min_{\mathbf{r} \in R}C_{s}(\theta_{s},r_{0})
+\sum_{i=1}^{n}C_{i}(\theta _{i},r_{i}).
\end{equation}
\end{definition}

Since the link cost functions $C_{i}(\theta _{i},r_{i})$ as well as $%
C_{s}(\theta_{s},r_{0})$ are all strictly increasing and strictly convex,
the socially optimal allocation $\mathbf{r}^*$ exists and is unique. The
conditions for specifying $\mathbf{r}^*$ can be obtained using the
Kuhn-Tucker conditions. Let $c_i(\theta_i, r_i) \triangleq {\partial {%
C_{i}(\theta _{i},r_{i})}}/ {\partial {r_{i}}}$ and $c_s(\theta_s, r_0)
\triangleq {\partial {C_{s}(\theta _{s},r_{0})}}/ {\partial {r_{0}}}$ denote
the marginal cost function of link $i$ and the marginal cost function of the
overflow link for source $s$, respectively.

For the case of an inelastic source, $\mathbf{r}^* = (0, r^*_1, \ldots,
r^*_n)$ is the socially optimal allocation if and only if for each $%
i=1,\ldots, n$,
\begin{equation}
c_i(\theta_i, r^*_i) = c^*~\text{if}~r^*_i > 0, \quad c_i(\theta_i, r^*_i) >
c^*~\text{if}~r^*_i = 0.  \label{eq:optcond}
\end{equation}
For the case of an elastic source, $\mathbf{r}^* = (r^*_0, r^*_1, \ldots,
r^*_n)$ is the socially optimal allocation if and only if~\eqref{eq:optcond}
holds and furthermore,
\begin{equation*}
c_s(\theta_s, r^*_0) = c^*~\text{if}~r^*_0 > 0, \quad c_s(\theta_s, r^*_0) >
c^*~\text{if}~r^*_0 = 0.
\end{equation*}












\subsection{{\protect\large \textsl{Game Structure}}}

\label{sec:general_game} Unlike the cooperative setting, in a network
consisting of selfish and strategic nodes, the source as well as the relays
will strategize to maximize their own utility, rather than work together to
minimize the overall network cost. Since forwarding traffic entails cost,
the relays will carry the source's traffic only if they are sufficiently
well compensated. The source, on the hand, wishes to have its traffic
forwarded at the smallest possible cost to itself. The natural setting in
which to carry out this game is one which allows for transfer payments which
accompany traffic allocations from the source to the respective relays.

In this work, we assume that the (maximum) source input rate $r_s$ and the
parameter $\theta_s$ are known to all nodes. As discussed above, the cost
function $C_{i}(\theta_{i},r_{i})$ for relay $i$ depends on the path quality
parameter or type $\theta_i$. In practical network settings, the value of
this type may be randomly fluctuating. For instance, in wireless
communication, the channel gain $g_{ij}$ fluctuates due to shadowing and
fading. In the Internet, the quality of a particular path may fluctuate
according to network congestion levels. Accordingly, we may assume that $%
\theta_i$ is randomly distributed according to distribution function $%
F_i(\theta_i)$. In practical network scenarios, the exact realization of $%
\theta_i$ is typically known only to relay $i$, and not to the source or to
the relays other than $i$. Thus, $\theta_i$ is \emph{private information} to
relay $i$. Nevertheless, the source and other relays may still have
knowledge of the distribution $F_i(\theta_i)$. For instance, a wireless
source or a relay $j \neq i$ may know the distribution of the channel gains
for relay $i$, but typically does not know the realization of those channel
gains. An Internet source or a path $j \neq i$ may know the distribution of
the congestion level on path $i$, but does not know the exact realization of
the congestion level.

In order for the source node to allocate its traffic intelligently in the
presence of incomplete information regarding the $\theta_i$'s, it needs to
observe some ``signal" from the relay nodes. This can be realized by having
the relay node send a signal according to the realization of its type to the
source.\footnote{%
One can also consider the possibility of the source sending a signal
according to its type $\theta_s$. However, since we assume $\theta_s$ is
known to all network nodes, we do not consider this possibility here.} Let $%
M_i$ be the set of signals for relay $i$, where $M_i$ is a subset of the set
of differentiable functions on $[0, r_s]$. The signal map for relay $i$ is
\begin{equation*}
s_i : \Theta _{i} \rightarrow M_{i},
\end{equation*}
where $\Theta _{i} \triangleq [\underline{\theta}_i, \overline{\theta}_i]$
and $s_i(\theta_i) = m_i(\cdot)$.

Given the signals $m_i(\cdot), i=1,\ldots,n$, the source decides on an
allocation of its traffic as well as a vector of transfer payments to the
relays. This allocation is called a \emph{contract}. Let $\mathbf{r}%
=(r_{0},r_{1},...,r_{n}) \in R$ denote the vector of traffic
rates in the network, where $r_0$ is the rate withheld by the source, and $%
r_i$ is the rate routed to relay $i, i=1, \ldots, n$. Note that for the
inelastic case, $r_0 = 0$. Now let $\mathbf{t}=(t_{1},t_{2},...,t_{n}) \in {%
\mathbb{R}}_+^n$ be the vector of transfer payments, where $t_i$ is the
transfer payment to relay $i$. Let $M \triangleq M_1 \times \cdots \times
M_n $ and $T \triangleq {\mathbb{R}}_+^n$.
Then the allocation map of the source node is
\begin{equation*}
g:M\rightarrow R\times T,
\end{equation*}
where $g(m_1(\cdot), \ldots, m_n(\cdot)) = (\mathbf{r}, \mathbf{t})$.

The above framework encompasses many forms of pricing games explored in
previous literature. For instance, in~\cite{Yufang06}, the relay signals are
simply charging functions $P_i(\cdot)$, and the transfer payments are
required to equal the charges demanded by the relays, i.e. $t_i = P_i(r_i)$.

The signal maps of the relays along with the allocation map of the source
realize a corresponding network allocation map
\begin{equation*}
f:\Theta \rightarrow R\times T,
\end{equation*}%
where $f(\theta _{1},\ldots ,\theta _{n})=g(s_{1}(\theta _{1}),\ldots
,s_{n}(\theta _{n}))=(\mathbf{r},\mathbf{t})$.

In the game with incomplete information corresponding to the above setting,
the utility of the source is given by
\begin{equation*}
{U}_s(\theta_s, g(s_1(\theta_1), \ldots, s_n(\theta_n)) = W_s(r_s) -
C_s(\theta_s, r_0) - \sum_{i=1}^n t_i.
\end{equation*}
The utility of relay $i$ is given by
\begin{equation*}
U_i(\theta_i, g(s_1(\theta_1), \ldots, s_n(\theta_n))) = t_i - C_i(\theta_i,
r_i).
\end{equation*}

The game with incomplete information proceeds as follows. First, each relay $%
i$ observes its own private information $\theta_i$. Second, the source
provides a contract for the relay nodes. The contract announces the source
allocation rule $g:M\rightarrow R\times T$. Third, the relays simultaneously
decide to either accept or reject the contract. If a given relay accepts the
contract, then it will participate in the game which follows. Otherwise, the
relay quits and receives zero utility.\footnote{Note that the relays which
quit can simply be left out of the game formulation.  Thus, without loss of generality,
we assume for the rest of the paper that the source plays the game in a manner
which gives non-negative expected utility to all relays, so that all relays stay
in the game.}  Fourth and finally, the relay nodes
simultaneously send their signals to the source, and the source allocates
rates and transfer payments according to the announced $g$.

In the following, we give the formal definition of the Bayesian Nash
equilibrium corresponds to the game with incomplete information described
above. Let ${\mathbf{\theta}} \triangleq (\theta_1, \ldots, \theta_n)$, ${%
\mathbf{\theta}}_{-i} \triangleq (\theta_j)_{j \neq i}$, and $s_{-i}({%
\mathbf{\theta}}_{-i}) \triangleq (s_j(\theta_j))_{j \neq i}$.

\begin{definition}
A Bayesian Nash Equilibrium of the above game is a set of strategies $\{s_1,
\ldots, s_n, g\}$ satisfying \newcounter{BNE}
\begin{list}{\bfseries\upshape \arabic{BNE}.}
{\usecounter{BNE}}
\item for each relay node $i$ and every feasible $\widetilde{s_i}: \Theta_i \rightarrow M_i$,
\begin{eqnarray}
E_{\theta_{-i}} \left\{U_i(\theta_i, g(s_i(\theta_i), s_{-i}(\theta_{-i}))) \right\} \nonumber\\
 \ge E_{\theta_{-i}} \left\{ U_i(\theta_i, g(\widetilde{s_i}(\theta_i), s_{-i}(\theta_{-i})) )\right\},
 \label{eq:relayopt}
\end{eqnarray}
\item for every feasible $\widetilde{g}: M \rightarrow R \times T$,
\begin{equation}
E_{\mathbf{\theta}} \left\{U_s(\theta_s, g(s(\mathbf{\theta})))\right\} \ge
E_{\mathbf{\theta}} \left\{ U_s(\theta_s, \widetilde{g}(s(\mathbf{\theta})) ) \right\}.
\label{eq:sourceopt}
\end{equation}
\end{list}
\end{definition}


\section{{\protect\large \textsl{Games with Limited Source Bargaining Power}}%
}

We first consider a specific instance of the general game described in
Section~\ref{sec:general_game} in which the source has limited bargaining
power. In this case, the source can only react passively to the relays'
signals. Specifically, the transfer payment from the source to any given
relay must equal the relay's signal function evaluated at the traffic rates
routed to the relay. That is, the source allocation rule is given by $%
g(m_1(\cdot), \ldots, m_n(\cdot)) = (\mathbf{r}, \mathbf{t})$, where
\begin{eqnarray}
\mathbf{r} & \in & \arg \max_{\mathbf{r}^{\prime }\in
R}W_{s}(\theta_s,r_{s})-C_{s}(\theta _{s},r^{\prime }_{0})-\displaystyle%
\sum\limits_{i = 1}^n m_{i}(r_{i}^{\prime })  \label{eq:sourcenobargain1} \\
t_i & = & m_i(r_i), \quad i = 1, \ldots, n.  \label{eq:sourcenobargain2}
\end{eqnarray}
Effectively, the relays' signal functions act as charging functions, and the
transfer payments must correspond to the relays' charges. The source can
only allocate its traffic to minimize the cost of withheld traffic plus the
total charges paid to the relays. In this case, the game can be considered
to be a \emph{pricing game}.


\subsection{{\protect\large \textsl{Pricing Game with Complete Information}}}

In this section, we consider the specific pricing game with \emph{complete
information} where the source has limited bargaining power and the vector of
relay types $\theta = (\theta_1, \ldots, \theta_n)$ is known to all nodes in
the network. Note that this is degenerate version of the game considered in
Section~\ref{sec:general_game} where the prior distribution on the type of
relay $i$ available to all nodes is given by the distribution function $%
F_i(x) = 0$ for $x < \theta_i$ and $F_i(x) = 1$ for $x \geq \theta_i$, where
$\theta_i$ is the realization of relay $i$'s type.

Since the allocation rule of the source is fixed by~(\ref%
{eq:sourcenobargain1})-(\ref{eq:sourcenobargain2}), the knowledge of $\theta$
cannot cause the source to adjust its allocation rule accordingly. Thus,
knowledge of $\theta$ is not useful to the source due to its lack of
bargaining power. Also, due to the degenerate prior distribution on $%
\theta_i $, we need only consider the usual concept of Nash equilibrium
here. We now show that in fact all the Nash equilibria in this complete
information pricing game are efficient.

\begin{theorem}
In the pricing game with complete information, Nash equilibria exist, and
all Nash equilibria are efficient. Moreover, there exists an efficient Nash
equilibrium in which each relay uses a linear charging function.
\end{theorem}

\begin{proof}
We focus on the case for inelastic sources. The elastic case can be
similarly handled. Since $\theta = (\theta_1, \ldots, \theta_n)$ is known to
all nodes in the network, we suppress the dependence of various quantities
on $\theta$. In this game with limited source bargaining power, the relays'
signals represent charging functions. Let $B_i(r_i)$ be the charge required
by relay $i$ for forwarding traffic of rate $r_i$, and let $b_i(r_i)
\triangleq B^{\prime }_i(r_i)$ be the marginal charging function, or pricing
function. Let $C_i(r_i)$ and $c_i(r_i)$ be cost function and marginal cost
function for relay $i$, respectively.

Let the (unique) socially optimal allocation be ${\mathbf{r}}^* = (r_1^*,
r_2^*, ..., r_n^*)$. Suppose that there exists a Nash Equilibrium with
charging functions $B_i(r_i)$ and corresponding rate allocation ${\mathbf{r}}
= (r_1, r_2, ..., r_n) \neq {\mathbf{r}}^*$. With a possible re-ordering of
the relay indices, we may assume that $r_i>r_i^*$ for $i<k_1$, $r_i=r_i^*$
for $k_1 \le i < k_2$, and $r_i<r_i^*$ for $i \ge k_2$. As ${\mathbf{r}}
\neq {\mathbf{r}}^*$ and both must sum to $r_s$, $k_1>1$ and $k_2<n$.

Since ${\mathbf{r}}^*$ is the unique socially optimal allocation, from the
optimality conditions, we have
\begin{equation}
c_i(r_i^*)= c^*~\text{if}~r^*_i > 0, \quad c_i(r_i^*)> c^*~\text{if}~ r^*_i
= 0.
\end{equation}
where $c^*$ is the optimal marginal cost. Now by the strict convexity of $%
C_i(r_i)$,
\begin{equation}  \label{Maginal1}
\left\{
\begin{array}{ll}
c_i(r)> c^*~\text{for all}~r \in [r^*_i, r_i)~\text{if}~i<k_1 &  \\
c_i(r)< c^*~\text{for all}~r \in (r_i, r^*_i]~\text{if}~i>k_2 &
\end{array}
\right.
\end{equation}
The profit of relay $i$ for $i<k_1$ is
\begin{equation}
\int_0^{r_i^*} b_i(r)- c_i(r) dr + \int_{r_i^*}^{r_i} b_i(r) - c_i(r) dr.
\end{equation}
Since we are at a Nash equilibrium, for all $i < k_1$ and for any $%
0<\delta<r_i-r_i^*$, $\int_{r_i - \delta}^{r_i} b_i(r) - c_i(r) dr \geq 0$.
For otherwise, relay $i < k_1$ will deviate to another charging function
which is extremely high from $r_i^*$ to $r_i$, so as not to take the extra
traffic $r_i-r_i^*$. Now choose $\epsilon < min_{i: i < k_1~\text{or}~i \geq
k_2}|r_i-r_i^*|$. Let
\begin{equation}
\left\{
\begin{array}{ll}
m & \in \quad \arg\max_{1\le i < k_1} \int_{r_i-\epsilon}^{r_i} c_i(r) dr \\
l & \in \quad \arg\min_{k_2\le i \le n} \int_{r_i}^{r_i+\epsilon} c_i(r) dr%
\end{array}
\right.
\end{equation}
By~(\ref{Maginal1}),
\begin{equation}
\int_{r_{m}-\epsilon}^{r_m} c_m(r) dr > \int_{r_{l}}^{r_{l}+\epsilon} c_l(r)
dr
\end{equation}
However, since $\int_{r_{m}-\epsilon}^{r_{m}} b_{m}(r) - c_{m}(r) dr \geq 0 $%
, there exists a charging function $\tilde{B}_{l}(r)$ for relay $l$ such
that $\tilde{B}_{l}(r)$ equals $B_{l}(r)$ from 0 to $r_{l}$, but
\begin{equation}
\begin{array}{ll}
\int_{r_{m}-\epsilon}^{r_{m}} b_{m}(r) dr & \geq \quad
\int_{r_{m}-\epsilon}^{r_{m}} c_m(r) dr \\
& > \quad \int_{r_l}^{r_l+\epsilon} \tilde{b}_l(r) dr \quad > \quad
\int_{r_l}^{r_l+\epsilon} c_l(r) dr.%
\end{array}%
\end{equation}
Thus if relay $l$ uses $\tilde{B}_l(r)$, then in order to maximize its
profit, the source will switch an $\epsilon$ amount of traffic from relay $m$
to relay $l$. Thus, relay $l$ can deviate to $\tilde{B}_l(r)$ and get a
higher profit, contradicting our assumption of being at a Nash equilibrium.

The existence of an efficient Nash equilibrium in which relays use linear
charging functions has been demonstrated in~\cite{Yufang06}, completing the
proof.
\end{proof}

\subsection{{\protect\large \textsl{\ Pricing Game with Incomplete
Information}}}

When the source and the relays $j\neq i$ cannot observe the type $\theta
_{i} $ of relay $i$, the source and the relays must content themselves with
maximizing their expected profits. In this situation, the characterization
of Bayesian Nash Equilibria for general nonlinear charging functions is very
difficult. We limit our discussion to the case where relays bid \emph{linear}
charging functions, i.e. $B_{i}(\theta _{i},r_{i})=p_{i}(\theta _{i})r_{i}$,
where the price $p_{i}(\theta _{i})$ per unit traffic depends on the type $%
\theta _{i}$. Let $w_{i}\triangleq p_{i}^{-1}$ be the inverse function of $%
p_{i}$ such that $\theta _{i}=w_{i}(p_{i}(\theta _{i}))$. We assume that the
density $f_{i}(\theta _{i})$ is positive over $\Theta _{i}=[\underline{%
\theta }_{i},\overline{\theta }_{i}]$.

We prove the following theorem.

\begin{theorem}
If the source is inelastic, in any Bayesian Nash Equilibrium, the price
function satisfies the following differential equations:
\begin{eqnarray}
\frac{dw_{i}(p_{i})}{dp_{i}} &=&\frac{F_{i}(w_{i}(p_{i}))}{%
(n-1)f_{i}(w_{i}(p_{i}))}\Bigg\{\frac{-(n-2)r_{s}}{%
p_{i}r_{s}-C_{i}(w_{i}(p_{i}),r_{s})}  \notag \\
&&+\displaystyle\sum\limits_{j\neq i}\frac{r_{s}}{%
p_{i}r_{s}-C_{j}(w_{j}(p_{i}),r_{s})}\Bigg\},  \label{Inelastic Equation}
\end{eqnarray}
where $p_{i}(\theta _{i})$ is given by the inverse of $w_{i}(p_{i})$.

In particular, in the symmetric situation where $F_{i}(\theta _{i})=F(\theta
_{i})$ and $C_{i}(\theta _{i},r_{i})=C(\theta _{i},r_{i})$ for all $i$, the
Bayesian Nash Equilibrium satisfies:
\begin{equation}
p_{i}(\theta _{i})=\frac{1}{r_{s}}\left\{ C(\theta _{i},r_{s})-\frac{\int_{%
\underline{\theta }}^{\theta _{i}}F(\theta )^{n-1}\frac{\partial C(\theta
,r_{s})}{\partial \theta }d\theta }{F(\theta _{i})^{n-1}}\right\}.
\end{equation}
\end{theorem}

\begin{proof}
By an argument similar to that in \cite{Lebrun99}, $p_{i}(\theta _{i})$
and $w_{i}(p_{i})$ are both strictly decreasing functions. Since the
charging functions are linear, the source will always allocate all its
traffic to the relay proposing the lowest price.\footnote{%
Note that since $\theta _{i}$ are continuous random variables and $%
p_{i}(\theta _{i})$ are strictly decreasing functions, the probability that
there are any ties in the relay prices is zero.} Given the other relays'
pricing strategies $w_{j}(p_{j}),j\neq i$, the probability that relay $i$
proposes the lowest price is given by%
\begin{eqnarray*}
\Pr \{p_{i} &<&p_{j}\text{ for all }j\neq i\}=\Pr \{\theta _{j}<w_{j}(p_{i})%
\text{ for all }j\neq i\} \\
&=&\displaystyle\prod\limits_{j\neq i}F_{j}(w_{j}(p_{i}))
\end{eqnarray*}
For each given private type $\theta _{i}$, relay $i$ wishes to choose its
price $p_{i}$ to maximize the expected profit%
\begin{eqnarray}
\pi _{i}(\theta _{i},p_{i}) &=&\Pr \{p_{i}<p_{j}\text{ for all }j\neq
i\}(p_{i}r_{s}-C_{i}(\theta _{i},r_{s}))  \notag \\
&=&\displaystyle\prod\limits_{j\neq
i}F_{j}(w_{j}(p_{i}))(p_{i}r_{s}-C_{i}(\theta _{i},r_{s}))
\end{eqnarray}
In order to maximize $\pi _{i}(\theta _{i},p_{i})$, the first-order
condition must be hold:
\begin{eqnarray}
\frac{\partial \ln \pi _{i}(\theta _{i},p_{i})}{\partial p_{i}} &=&%
\displaystyle\sum\limits_{j\neq i}\frac{1}{F_{j}(w_{j}(p_{i}))}%
f_{j}(w_{j}(p_{i}))\frac{dw_{j}(p_{i})}{dp_{i}}  \notag \\
&&+\frac{r_{s}}{p_{i}r_{s}-C_{i}(\theta _{i},r_{s})} \\
&=&0  \notag
\end{eqnarray}
After some algebra, we obtain (\ref{Inelastic Equation}).

We now focus on the symmetric situation for an inelastic source, where $%
F_{i}(\theta _{i})=F(\theta _{i})$ and $C_{i}(\theta _{i},r_{i})=C(\theta
_{i},r_{i})$ for all $i$. First, using an argument similar to that in~\cite%
{Lebrun99}, all the relay nodes should have the same pricing strategy $%
p(\theta _{i})$ and $w(p_{i})$. Thus, the expected profit for relay $i$ is
\begin{equation}
\pi _{i}(\theta _{i},p_{i})=F(w(p_{i}))^{n-1}(p_{i}r_{s}-C(\theta
_{i},r_{s})).  \label{eq:expectedprofit}
\end{equation}

Let the value function for relay $i$ (the maximum profit for relay $i$ given
type $\theta _{i}$ by choosing the optimal $p_{i}(\theta _{i})$) be $%
v_{i}(\theta _{i})\triangleq \max_{p_{i}}\pi _{i}(\theta _{i},p_{i})$. By
the envelope theorem,
\begin{eqnarray}
\frac{dv_{i}(\theta _{i})}{d\theta _{i}} &=&\frac{\partial
\{F(w(p_{i}))^{n-1}(p_{i}r_{s}-C(\theta _{i},r_{s}))\}}{\partial \theta _{i}}%
\Bigg|_{p_{i}=p_{i}(\theta _{i})}  \notag \\
&=&-F(w(p_{i}))^{n-1}\frac{\partial C(\theta _{i},r_{s})}{\partial \theta
_{i}}  \notag \\
&=&-F(\theta _{i})^{n-1}\frac{\partial C(\theta _{i},r_{s})}{\partial \theta
_{i}}
\end{eqnarray}

Since $p(\theta _{i})$ is decreasing, the lowest type player must win zero
expected profit, i.e., $v_{i}(\underline{\theta })=0$. Thus,
\begin{equation}
v_{i}(\theta _{i})=\int_{\underline{\theta }}^{\theta _{i}}-F(\theta )^{n-1}%
\frac{\partial C(\theta ,r_{s})}{\partial \theta }d\theta.
\label{eq:valuefunction}
\end{equation}

We now use~\eqref{eq:expectedprofit} and~\eqref{eq:valuefunction} to solve
for the optimal pricing function:
\begin{eqnarray}
p_{i}(\theta _{i}) &=&p(\theta _{i})  \notag \\
&=&\frac{1}{r_{s}}\left\{ \frac{v_{i}(\theta _{i})}{F(w(p_{i}))^{n-1}}%
+C(\theta _{i},r_{s})\right\} \\
&=&\frac{1}{r_{s}}\left\{ C(\theta _{i},r_{s})-\frac{\int_{\underline{\theta
}}^{\theta _{i}}F(\theta )^{n-1}\frac{\partial C(\theta ,r_{s})}{\partial
\theta }d\theta }{F(\theta _{i})^{n-1}}\right\}  \notag
\end{eqnarray}
\end{proof}

The case of an elastic source can be treated in a similar way.  We omit the proof
here and simply state the result.

\begin{theorem}
If the source is elastic, in any Bayesian Nash Equilibrium, the price
function satisfies the following differential equations:
\begin{eqnarray}
\frac{dw_{i}(p_{i})}{dp_{i}} &=&\frac{F_{i}(w_{i}(p_{i}))}{%
(n-1)f_{i}(w_{i}(p_{i}))}\Bigg\{\frac{-(n-2)(r_{s}-r_{0}(p_{i})}{%
p_{i}r_{s}-C_{i}(w_{i}(p),r_{s}-r_{0})}  \notag \\
&&-\frac{-\frac{dr_{0}(p_{i})}{dp_{i}}(p_{i}-\frac{\partial C_{i}(\theta
_{i},r_{s}-r_{0})}{\partial (r_{s}-r_{0})}))}{%
p_{i}r_{s}-C_{i}(w_{i}(p),r_{s}-r_{0})}  \label{Elastic Equation} \\
&&+\displaystyle\sum\limits_{j\neq i}\frac{r_{s}+\frac{dr_{s}(p_{i})}{dp_{i}}%
(p_{i}-\frac{\partial C_{i}(\theta _{i},r_{s}-r_{0})}{\partial (r_{s}-r_{0})}%
)}{p_{i}r_{s}-C_{j}(w_{j}(p),r_{s}-r_{0})}\Bigg\}  \notag
\end{eqnarray}
where $p_{i}(\theta _{i})$ is given by the inverse of $w_{i}(p_{i})$.
\end{theorem}
\subsection{{\protect\large \textsl{Efficiency Analysis}}}

In this section, we measure the inefficiency introduced by the pricing game
with incomplete information. We shall use the useful measure \emph{price of
anarchy}, defined for each given type vector $\theta $.

\begin{definition}
The price of anarchy $\rho(\theta)$ for a given type vector $\theta$ in the
incomplete information game is
\begin{equation}
\rho(\theta)= \frac{\max_{r_i \in R^E}\sum_i C_i(\theta_i, r_i)}{\min_{r_i
\in R}\sum_{i} C_i(\theta_i, r_i)}
\end{equation}
where $R^E$ is the set of all traffic allocations corresponding to Bayesian
Nash equilibria, and $R$ is the set of all feasible traffic allocations.
\end{definition}

We shall focus on the case of an inelastic source. The elastic source case
is similar. We consider the symmetric situation where $F_{i}(\theta
_{i})=F(\theta _{i})$ and $C_{i}(\theta _{i},r_{i})=C(\theta _{i},r_{i})$
for all $i$. In the case where all relays bid linear charging functions, the
highest type relay will receive all the traffic. Here, the price of anarchy
is determined by
\begin{equation}
\rho (\theta )=\frac{C(\max_{i\in I}\theta _{i},r_{s})}{\min_{{\mathbf{r}}
\in R}\displaystyle\sum\limits_{i}C(\theta _{i},r_{i})}
\end{equation}
We develop the following bound on $\rho (\theta)$.

\begin{theorem}\label{Efficiency1}
In the symmetric linear pricing game with incomplete information, if the
marginal cost function $c(\theta_i,r_i) = \frac{\partial C(\theta_i, r_i)}{%
\partial r_i}$ is concave, then $\rho (\theta) \leq n$, where $n$ is the
number of relays, with equality if and only if $c(\theta_i,r_i)$ is linear 
in $r_i$ and the relay types $\theta_i$ are all the same.
\label{thm:pa1}
\end{theorem}

\begin{proof}
Let $(r_{i}^{\ast })_{i=1}^{n}=(a_{i}r_{s})_{i=1}^{n}$ be the socially
optimal allocation for a given type realization $\theta $, where $%
\sum_{i=1}^{n}a_{i}=1$ and $a_{i}\geq 0$ for all $i$. Thus the optimal cost
is
\begin{equation}
C^{\ast }=\sum_{i=1}^{n}\int_{0}^{a_{i}r_{s}}c(\theta _{i},r_{i})dr_{i}
\label{eq:optcost}
\end{equation}%
Since $c(\theta _{i},r_{i})$ is concave, it can be shown that $%
\int_{0}^{a_{i}r_{s}}c(\theta _{i},r_{i})dr_{i}\geq
a_{i}^{2}\int_{0}^{r_{s}}c(\theta _{i},r_{i})dr_{i}$, where equality holds
if and only if $c(\theta _{i},r_{i})$ is linear in $r_{i}$. Thus we have
\begin{equation}
C^{\ast }\geq \sum_{i=1}^{n}a_{i}^{2}\int_{0}^{r_{s}}c(\theta
_{i},r_{i})dr_{i}  \label{C*}
\end{equation}%

Therefore,


\begin{eqnarray}
\rho (\theta ) &=& \frac{\int_{0}^{r_{s}}c(\max_{i}\theta _{i},r_{i})dr_{i}}{%
\sum_{i=1}^{n}\int_{0}^{a_{i}r_{s}}c(\theta _{i},r_{i})dr_{i}} \nonumber\\
&\leq& \frac{\int_{0}^{r_{s}}c(\max_{i}\theta _{i},r_{i})dr_{i}}{%
\sum_{i=1}^{n}a_{i}^{2}\int_{0}^{r_{s}}c(\theta _{i},r_{i})dr_{i}} \\
&\leq &\frac{\int_{0}^{r_{s}}c(\max_{i}\theta _{i},r_{i})dr_{i}}{%
\sum_{i=1}^{n}a_{i}^{2}\int_{0}^{r_{s}}c(\max_{i}\theta _{i},r_{i})dr_{i}}
=\frac{1}{\sum_{i=1}^{n}a_{i}^{2}}\leq n \nonumber
\end{eqnarray}
where the second inequality follows from the assumption that
${\partial ^{2}{C_i(\theta
_{i},r_{i})}}/{\partial {\theta _{i}}\partial{r_{i}}}\leq 0$.
Equality obtains in all three previous inequalities if $c(\theta_i,r_i)$ is linear 
in $r_i$ and the relay types $\theta_i$ are all the same.
\end{proof}

Next, we give a general bound on the price of anarchy for all cost functions
satisfying our assumptions in Section~\ref{model}. 

\begin{theorem}\label{Efficiency2}
In the symmetric linear pricing game with incomplete information, let the
support set for each $\theta _{i}$ be $\Theta \triangleq \lbrack \underline{%
\theta },\overline{\theta }]$. If the marginal cost function $c(\theta
_{i},r_{i})=\frac{\partial C(\theta _{i},r_{i})}{\partial r_{i}}$ satisfies $%
\frac{c(\underline{\theta },r_{s})}{c(\overline{\theta },0)}\leq k$ for some
constant $k$, then $\rho (\theta )\leq k$.
\end{theorem}

\begin{proof}  Since $C(\theta _{i},r_{i})$ is convex in $r_i$ and 
${\partial ^{2}{C_i(\theta
_{i},r_{i})}}/{\partial {\theta _{i}}\partial{r_{i}}}\leq 0$ by assumption,
$c(\theta _{i},r_{i}) \geq c(\theta _{i},0) \geq c(\max_i \theta _{i},r_{i})
\geq c(\overline{\theta},0)$.   Also $c(\max_i \theta _{i},r_{i}) \leq 
c(\underline{\theta}, r_s)$.   Thus,  $\sum_{i=1}^{n}\int_{0}^{a_{i}r_{s}}c(\theta _{i},r_{i})dr_{i}$
$\geq c(\overline{\theta},0)r_s$, and 
$\int_{0}^{r_{s}}c(\max_{i}\theta _{i},r_{i})dr_{i}$ $\leq c(\underline{\theta}, r_s) r_s$. 
The result follows.
\end{proof}

Recall our result that all Nash equilibria in the complete information
pricing game are efficient, including any which results from linear pricing.
Thus, we see that incomplete information can introduce inefficiencies. The
main insight is that in an incomplete information pricing game, the relays
cannot calculate the socially optimal traffic allocation due to the lack of
information regarding types. Therefore, the relays cannot bid the marginal
cost at the socially optimal outcome as the price, Thus, the game cannot
reach an efficient Nash Equilibrium.

Although Bayesian Nash Equilibria are not efficient in the symmetric linear
pricing game with incomplete information, they satisfy an asymptotic
efficient property: the outcome of the Bayesian Nash Equilibrium when $r_{s}$
goes to zero is efficient. To see this, note that by \cite{Lebrun99}, all
the relay pricing functions in the symmetric case are the same and
decreasing. Thus the highest type relay will always get all the traffic.
When $r_{s}$ goes to zero, the efficient allocation also allocates all the
traffic to the highest type relay. Thus, in the symmetric case, a Bayesian
Nash Equilibrium is efficient when $r_{s}$ goes to zero.

We now show, however, that in the asymmetric linear pricing game with
incomplete information, the Bayesian Nash Equilibria are not efficient even
when $r_{s}$ goes to zero. We focus on the case of two relays, where the
cost functions of the relays are identical, but the distributions of the
types are different. Using Theorem 2, we obtain the following differential
equations:
\begin{eqnarray*}
\frac{dw_{1}}{dp} &=&-\frac{r_{s}F_{1}(w_{1}(p))}{%
(pr_{s}-C(w_{2},r_{s}))f_{1}(w_{1})} \\
\frac{dw_{2}}{dp} &=&-\frac{r_{s}F_{2}(w_{2}(p))}{%
(pr_{s}-C(w_{1},r_{s}))f_{2}(w_{2})}
\end{eqnarray*}
Explicitly solving for the solution is difficult, but we can observe some
properties of the solution. First, we must have
\begin{equation*}
p_{1}(\overline{\theta })=p_{2}(\overline{\theta })=p_{\min }.
\end{equation*}
This is because if the relay prices for the highest type are not the same,
then the relay with the higher price will lower its price to increase its
probability of winning the game, thus increasing the expected revenue. From
the differential equations, we obtain
\begin{equation}
w_{1}(p)=\overline{\theta }-\int_{p_{\min }}^{p}\frac{r_{s}F_{1}(w_{1}(p))}{%
(pr_{s}-C(w_{2}(p),r_{s}))f_{1}(w_{1}(p))}dp
\end{equation}
\begin{equation}
w_{2}(p)=\overline{\theta }-\int_{p_{\min }}^{p}\frac{r_{s}F_{2}(w_{2}(p))}{%
(pr_{s}-C(w_{1}(p),r_{s}))f_{2}(w_{2}(p))}dp
\end{equation}

For a given $p$, let $\theta_1$ and $\theta_2$ be such that $p_{1}(\theta
_{1})=p_{2}(\theta _{2})=p$. From the above equations, it is clear that $%
w_{1}(p) \neq w_{2}(p)$, i.e. $\theta _{1}\neq \theta _{2}$. Therefore, we
have a situation where two relays with different type propose the same
price. When this realization occurs, the highest type relay does not carry
all the traffic, even when $r_{s}$ goes to zero. Thus, in the asymmetric
case, the Bayesian Nash Equilibrium is not asymptotic efficient as $r_{s}$
goes to zero.

\section{{\protect\large \textsl{Games with Full Source Bargaining Power}}}

In the discussion thus far, the source has limited bargaining power, and
passively reacts to the relays' signals, which are equivalent to charging
functions. The source can only allocate its traffic to minimize its cost in
withheld traffic plus the total transfer payment to the relays. In this
section, we examine the scenario where the source has full bargaining power,
in the sense that the contract announced by the source is not limited to the
one described in~\eqref{eq:sourcenobargain1}-\eqref{eq:sourcenobargain2}. We first
investigate the (Bayesian) Nash equilibria which can result from games
with source bargaining power in the case of complete information.   Here,
we show that all (Bayesian) Nash equilibria are efficient.  Then,
we proceed to the case of incomplete information, and characterize the potential
inefficiencies associated with that case.

\subsection{{\protect\large \textsl{Games with Complete Information}}}

In a game with source bargaining power and complete information, the source
can observe the type vector $\theta = (\theta_1, \ldots, \theta_n)$ of the
relays, and then design the allocation map $g$ according to $\theta$. Since
the type $\theta_i$ is not private to relay $i$, relay $i$ cannot manipulate
this information in designing its signalling strategy $s_i$. Since the
source can observe $\theta$, it can effectively ignore the strategies of the
relays in designing $g$. Nevertheless, the source needs to ensure that the
relays will accept its proposed contract and stay in the game. The latter
will hold as long as $U_i(\theta_i, g(\theta)) = t_i - C_i(\theta_i, r_i)
\geq 0$ for all $i$. That is, all relays receive non-negative utility by
accepting the contract proposed by the source, and therefore are willing to
participate in the game.

\begin{lemma}
In any (Bayesian) Nash Equilibrium of the complete information game with
source bargaining power, all relays receive zero utility. \label{lemma:zero}
\end{lemma}

\begin{proof}
Suppose that there exists a (Bayesian) Nash Equilibrium where the source
allocation rule
\begin{equation*}
g(m_1(\cdot), ... , m_n(\cdot)) = ({\mathbf{r}}, {\mathbf{t}})
\end{equation*}
is such that $U_i(\theta_i, r_i, t_i) = t_i - C_i(\theta_i, r_i) > 0$ for
some $i$. Since the source can observe $\theta$, it could select another
allocation rule $g^{\prime }(m_1(\cdot), ... , m_n(\cdot)) = ({\mathbf{r}}%
^{\prime }, {\mathbf{t}}^{\prime })$ such that
\begin{equation*}
r^{\prime }_i = r_i, i =1,\ldots,n; \quad t^{\prime }_i = t_i-\epsilon, \;\;
t^{\prime }_j = t_j~\text{for all}~j \neq i
\end{equation*}
where $\epsilon$ is small enough so that $t^{\prime }_i - C_i(\theta_i,
r^{\prime }_i) > 0$. Note that the set of relays which would opt to accept
contract $g$ and stay in the game is the same as the set for contract $%
g^{\prime }$. On the other hand, by shifting its allocation rule from $g$ to
$g^{\prime }$, the source has strictly decreased its total transfer payment,
while keeping the same traffic allocation. Thus, the source's utility is
strictly increased. This contradicts our assumption of being at a Nash
equilibrium.
\end{proof}

\begin{theorem}
In the complete information game with source bargaining power, all
(Bayesian) Nash equilibria are efficient. \label{thm:completebargain}
\end{theorem}

\begin{proof}
At any Nash equilibrium, the source maximizes its utility
\begin{equation*}
{U}_s(\theta_s, g(s_1(\theta_1), \ldots, s_n(\theta_n)) = W_s(\theta_s, r_s)
- C_s(\theta_s, r_0) - \sum_{i=1}^n t_i.
\end{equation*}
By Lemma~\ref{lemma:zero}, at the equilibrium, we have $t_i = C_i(\theta_i,
r_i)$ for all $i$. Thus, the traffic allocation by the source minimizes $%
C_s(\theta_s, r_0) + \sum_{i=1}^n C_i(\theta_i, r_i)$, and therefore the
equilibrium is efficient.
\end{proof}

Using Lemma~\ref{lemma:zero} and Theorem~\ref{thm:completebargain}, we can
easily solve for the Nash equilibrium of the complete information game with
source bargaining power. By Theorem~\ref{thm:completebargain}, the source
allocation rule at the equilibrium may be obtained by solving for the
socially optimal traffic allocation ${\mathbf{r}}^*$, where ${\mathbf{r}}^*
= \arg \max_{{\mathbf{r}} \in R} C_s(\theta_s, r_0) + \sum_{i=1}^n
C_i(\theta_i, r_i)$. As noted in Section~\ref{sec:efficient}, due to the
strict convexity of the optimization problem, ${\mathbf{r}}^*$ exists and is
unique. By Lemma~\ref{lemma:zero}, at the equilibrium, the transfer payment $%
t_i = C_i(\theta_i, r^*_i)$ for every $i = 1, \ldots, n$.\footnote{%
Recall that $C_i(\theta_i, 0) = 0.$} For the relays, any feasible signal map
$s_i$ may be chosen for the equilibrium.

To see why this constitutes an equilibrium, note the following sequence of
events in the game with source bargaining power. First, each relay $i$
observes its type $\theta_i$. Second, the source provides the contract $g: M
\rightarrow ({\mathbf{r}^*}, {\mathbf{t}})$, where ${\mathbf{r}}^*$ is the
socially optimal traffic allocation, and $t_i = C_i(\theta_i, r^*_i)$ for
every $i$. Note that $g$ is independent of the signals sent by the relays.
Third, the relays accept the mechanism because they each receive zero
utility, and therefore are indifferent with respect to carrying traffic or
not. Fourth and finally, the relay nodes will play signal map $s_i$ without
deviation, since the source allocation map is independent of the relays'
signals. Thus, the Nash equilibrium holds and is unique.

\subsection{{\protect\large \textsl{Games with Incomplete Information}}}

We now turn to the case that source cannot observe the type of relay. Thus
the relay nodes can manipulate their types in order to get more utility, and the
source can no longer design the allocation according to $\theta $.  As in
incomplete information games without bargaining power, the source must
maximize the expectation of profits according to the signals sent by relays.
The characterization of Bayesian Nash Equilibria for this case is very
difficult due to the complexity of the strategy set and the possible
behaviors of source and relays. Nevertheless, we devise a method for
characterizing outcomes corresponding to the Bayesian Nash Equilibria which
avoids the difficulty of calculating the the equilibria explicitly. We shall
do this in two steps. First, we show that if a resource allocation outcome
can be realized by a Bayesian Nash equilibrium for a game with source
bargaining in which every relay receives non-negative expected utility, then
there exists a ``truth telling" Bayesian Nash equilibrium that realizes the
outcome. Second, we show that the set of outcomes for the ``truth telling" Bayesian
Nash equilibria is included in the set of outcomes for the Nash equilibria for a \emph{%
complete information game}, in which the link cost functions are replaced by
a specified ``virtual cost functions."

\begin{definition}
A Bayesian Nash Equilibrium of the game with bargaining power is truth
telling if $M=\Theta $ \ and every relay node is willing to report their
true type to the source node.
\end{definition}

\begin{theorem}
\label{thm:revelation} If a resource allocation outcome $f$ can be realized
by a Bayesian Nash Equilibrium of the game with source bargaining power, in
which every relay receives non-negative expected utility, then there exists
a truth telling Bayesian Nash Equilibrium which realizes $f$.
\end{theorem}

\begin{proof}
Suppose there is a Bayesian Nash Equilibrium which realizes the allocation
outcome $f(\theta )$. By the definition of the Bayesian Nash Equilibrium, we
have~\eqref{eq:relayopt} and~\eqref{eq:sourceopt}. Now observe that by~%
\eqref{eq:relayopt}, we must have
\begin{eqnarray}
\theta _{i} &\in &\arg \max_{\widetilde{\theta _{i}}}E_{\theta _{-i}}\left\{
U_{i}(\theta _{i},g(s_{i}(\widetilde{\theta _{i}}),s_{-i}(\theta
_{-i})))\right\}
\end{eqnarray}
for all $i$.  Otherwise, if there exists some $\theta ^{\prime }$ such that $E_{\theta
_{-i}}\left\{ U_{i}(\theta _{i},g(s_{i}(\theta _{i}^{\prime }),s_{-i}(\theta
_{-i})))\right\}$ $>E_{\theta _{-i}}\left\{ U_{i}(\theta _{i},g(s_{i}(\theta
_{i}),s_{-i}(\theta _{-i})))\right\} $, then there is another strategy $%
s_{i}^{\prime }(\theta )$ satisfying $s_{i}^{\prime }(\theta
_{i})=s_{i}(\theta _{i}^{\prime })$ and $s_{i}^{\prime }(\theta
)=s_{i}(\theta )$ for all $\theta \neq \theta _{i}$, such that $E_{\theta
_{-i}}\left\{ U_{i}(\theta _{i},g(s_{i}^{\prime }(\theta _{i}),s_{-i}(\theta
_{-i})))\right\} $ $>E_{\theta _{-i}}\left\{ U_{i}(\theta
_{i},g(s_{i}(\theta _{i}),s_{-i}(\theta _{-i})))\right\} $, violating~%
\eqref{eq:relayopt}. Therefore, since $g(s_{1}(\theta _{1}),\ldots
,s_{n}(\theta _{n}))=f(\theta )$, we have
\begin{eqnarray}
\theta _{i} &\in &\arg \max_{\widetilde{\theta _{i}}\in \Theta
_{i}}E_{\theta _{-i}}\left\{ U_{i}(\theta _{i},f(\widetilde{\theta _{i}}%
,\theta _{-i}))\right\} \;\;\mbox{for~all~$i$} \;\;  \label{eq:truth1} \\
f &\in &\arg \max_{\widetilde{f}}E_{\theta }\left\{ U_{s}(\theta _{s},%
\widetilde{f}(\theta ))\right\} .  \label{eq:truth2}
\end{eqnarray}%
Thus, there exists a direct truth telling Bayesian Nash Equilibrium with the
outcome $f(\theta )$.
\end{proof}

Theorem \ref{thm:revelation} says that the set of outcomes corresponding to
Bayesian Nash Equilibria for the game with source bargaining power and
incomplete information is a subset of the outcomes corresponding to truth
telling Bayesian Nash Equilibria, in which each relay proposes its type
truthfully to the source, and the source optimally allocates rates according
to the relays' types. This finding simplifies our analysis considerably,
since we can now focus on the truth telling Bayesian Nash Equilibria in
order to bound the efficiency loss introduced by incomplete information in
games with source bargaining power.

We now investigate the outcomes which can be realized by truth telling
Bayesian Nash Equilibria. Notice that these equilibria correspond to the
solutions of the optimization problem given by~\eqref{eq:truth1} and~%
\eqref{eq:truth2}, in addition to the non-negative expected utility
constraint
\begin{equation}
E_{\theta _{-i}}\left\{ U_{i}(\theta _{i},r_{i})\right\} =E_{\theta
_{-i}}\left\{ t_{i}(\mathbf{\theta },\mathbf{r})-C_{i}(\theta
_{i},r_{i}(\theta ))\right\} \geq 0
\label{eq:nonnegative}
\end{equation}%
for all $i$, and feasibility constraint $\mathbf{r}\in R$.

\begin{theorem}
The set of solutions for the optimization problem defined by~%
\eqref{eq:truth1}-\eqref{eq:nonnegative} is included in the set of outcomes
corresponding to the Nash equilibria for the complete information game in
which the link cost functions $C_i(\theta_i, r_i)$ are replaced by
\begin{equation}
J_{i}(\theta _{i},r_{i})=C_{i}(\theta _{i},r_{i})-\frac{1-F_{i}(\theta _{i})%
}{f_{i}(\theta _{i})}\frac{\partial C_{i}(\theta _{i},r_{i})}{\partial
\theta _{i}}.  \label{Virtual Cost}
\end{equation}
\label{thm:virtual}
\end{theorem}

%
%
%

\begin{proof}
Please see the Appendix.
\end{proof}

We refer to the functions $J_i$ as \emph{virtual cost functions}. Note that
by Theorem~\ref{thm:completebargain}, all Nash equilibria corresponding to
games with complete information are efficient. Thus, the set of outcomes
corresponding to the Nash equilibria for the complete information game with
virtual link cost functions $J_i(\theta_i, r_i)$ is given by
\begin{equation*}
{\mathbf{r}}^{\prime }= \arg \min_{{\mathbf{r}} \in R} C_s(\theta_s, r_0) +
\sum_{i=1}^n J_i(\theta_i, r_i).
\end{equation*}

If $J_{i}(\theta _{i},r_{i})$ is strictly convex in $r_i$ for all $i$, then the optimization problem has
a unique solution.  For instance, if all the relays' types $\theta_i$ are
uniformly distributed on $[0,1]$, and the cost
functions are given by $\frac{1}{\theta_i}(e^{r_i}-1)$, then
$J_{i}(\theta
_{i},r_{i})=C_{i}(\theta _{i},r_{i})-\frac{1-F_{i}(\theta _{i})}{%
f_{i}(\theta _{i})}\frac{\partial C_{i}(\theta _{i},r_{i})}{\partial \theta
_{i}}=\frac{1}{\theta _{i}^{2}}(e^{r_{i}}-1)$, which is strictly
convex in $r_i$ and strictly decreasing in $\theta_i$.
In this case, if a Bayesian Nash Equilibrium of the game with
source bargaining power exists, then the
corresponding traffic allocation is the solution of the optimization problem.
In general, the set of traffic allocations corresponding to the Bayesian Nash Equilibria
(of the game with source bargaining power) is a subset of the solution set
for the optimization.  In the next section, we use this fact to bound the efficiency loss for games with incomplete
information.

\subsection{{\protect\large \textsl{Efficiency Analysis}}}

In this section, we bound the amount of inefficiency in the outcomes for
games with incomplete information. We focus on the inelastic scenario where $%
r_0 = 0$. Following~\cite{Roughgarden02}, define the price of anarchy for
type $\theta$ as:

\begin{equation}
\rho (\theta )=\frac{\max_{\mathbf{r}\in R^{E}} \sum\limits_{i}C_{i}(\theta
_{i},r_{i})}{\min_{\mathbf{r} \in R} \sum\limits_{i}C_{i}(\theta _{i},r_{i})}
\end{equation}
where $R^{E}$ is the set of all Bayesian Nash Equilibria for the game with
incomplete information.  Let $R^{J} \equiv \arg \min_{\mathbf{r} \in R}
\sum_{i}J_{i}(\theta _{i},r_{i})$. By Theorems~\ref{thm:revelation} and~\ref%
{thm:virtual}, we have $R^{E}\subseteq R^{J}$. Therefore,
\begin{equation}
\rho (\theta )\leq \frac{\max_{\mathbf{r} \in
R^{J}}\sum\limits_{i}C_{i}(\theta _{i},r_{i})}{\min_{\mathbf{r} \in
R}\sum\limits_{i}C_{i}(\theta _{i},r_{i})}
\end{equation}

Since the link cost functions are strictly convex, the socially optimal
allocation $\mathbf{r}^{\ast }$ are given by the necessary and sufficient
conditions in~\eqref{eq:optcond}. An allocation $\mathbf{r}^{\prime }$ in $%
R^{J}$ must satisfy the following necessary conditions: for all $i\in
\{1,\ldots ,n\}$ such that $r_{i}^{\prime }>0$,
\begin{eqnarray}
&&\frac{\partial C(\theta _{i},r_{i}^{\prime })}{\partial r_{i}}-\frac{%
1-F_{i}(\theta _{i})}{f_{i}(\theta _{i})}\frac{\partial ^{2}C(\theta
_{i},r_{i}^{\prime })}{\partial \theta _{i}\partial r_{i}}
\label{Virtual Optimality} \\
&\leq &\frac{\partial C(\theta _{j},r_{j}^{\prime })}{\partial r_{j}}-\frac{%
1-F_{j}(\theta _{j})}{f_{j}(\theta _{j})}\frac{\partial ^{2}C(\theta
_{i},r_{j}^{\prime })}{\partial \theta _{i}\partial r_{j}}\quad \text{for
all $j$}  \notag
\end{eqnarray}

We now bound the price of anarchy in the symmetric case.

\begin{theorem}
Consider the symmetric case where the link cost functions $C_i(\theta_i,
r_i) $ and the type distributions $F_i(\theta_i)$ are the same for all
relays. If (i) $J(\theta _{i},r_{i})$ is convex in $r_{i} $ and decreasing
in $\theta _{i}$, (ii) $X(\theta _{i},r_{i}) \equiv J(\theta
_{i},r_{i})-C(\theta _{i},r_{i})$ is concave in $r _{i}$, (iii) $\frac{%
\partial X(\theta _{i},r_{i})}{\partial \theta _{i}\partial r_{i}} \leq 0$,
then the price of anarchy $\rho(\theta)$ can be upper bounded as follows.

If
the marginal cost function $c(\theta_i,r_i) = \frac{\partial C(\theta_i, r_i)%
}{\partial r_i}$ is concave, then $\rho (\theta) \leq n$, where $n$ is the
number of relays (with equality if and only if $c(\theta_i,r_i)$ is linear 
in $r_i$ and the relay types $\theta_i$ are all the same). 
If $\frac{c(\underline{\theta },r_{s})}{c(\overline{\theta
},0)}\leq k$ for some constant $k$, then $\rho (\theta )\leq k$. \label%
{thm:anarchy}
\end{theorem}

Note that for the example where all the relays' types $\theta_i$ are
uniformly distributed on $[0,1]$ and the cost
functions are given by $\frac{1}{\theta_i}(e^{r_i}-1)$, the assumptions
of the Theorem are satisfied. 

\begin{proof}
Let $(r_{i}^{\prime })_{i\in I}\in \arg \min_{i}\sum_{i}J(\theta _{i},r_{i})$%
, and $(r_{i}^{\ast })$ be the efficient allocation. We first prove that if $%
\theta _{m}>\theta _{k}$, and $r_{m}+r_{k}=r_{mk}$ is fixed, $r_{m}^{\prime
}\geq r_{m}^{\ast }$. If for any $\theta _{k}<\theta _{m}$, $r_{k}^{\prime
}=0$, then the inequality immediately holds. We then consider the situation that
there exists some $\theta _{k}<\theta _{m}$ and $r_{k}^{\prime }>0$. Let $%
x(\theta _{i},r_{i})=\frac{\partial X(\theta _{i},r_{i})}{\partial r_{i}}$.
Thus for the optimal allocation, $r_{m}^{\ast }>r_{k}^{\ast }>0$ and $%
r_{m}^{\ast }+r_{k}^{\ast }=r_{mk}$. As $x(\theta _{i},r_{i})$ is decreasing
in $\theta _{i}$, we have $x(\theta _{m},r_{k}^{\ast }) \leq x(\theta
_{k},r_{k}^{\ast })$. As $x(\theta _{i},r_{i})$ is decreasing in $r_{i}$, $%
x(\theta _{m},r_{m}^{\ast }) \leq x(\theta _{m},r_{k}^{\ast })$. Thus $x(\theta
_{m},r_{m}^{\ast }) \leq x(\theta _{k},r_{k}^{\ast })$. As $c(\theta
_{m},r_{m}^{\ast })=c(\theta _{k},r_{k}^{\ast })$ , $c(\theta
_{m},r_{m}^{\ast })+x(\theta _{m},r_{m}^{\ast }) \leq c(\theta _{k},r_{k}^{\ast
})+x(\theta _{k},r_{k}^{\ast })$. By (\ref{Virtual Optimality}), $c(\theta
_{m},r_{m}^{\prime })+x(\theta _{m},r_{m}^{\prime }) \geq c(\theta
_{k},r_{k}^{\prime })+x(\theta _{k},r_{k}^{\prime })$, as $r_{k}^{\prime }>0$%
. As $r_{k}^{\ast }+r_{m}^{\ast }=r_{mk}=r_{k}^{\prime }+r_{m}^{\prime }$,
and virtual cost function is convex, we have $r_{m}^{\prime } \geq r_{m}^{\ast }$.

Now we prove that $\sum_{i}C(\theta _{i},r_{i}^{\prime })\leq
C(\max_{i}\theta _{i},r_{s})$. Without loss of generality, we assume that $%
\theta _{1}<\theta _{2}<...<\theta _{n}$.  Any situation where
some types are the same can be handled by modifying number of relays.  We have
\begin{eqnarray*}
&&\min_{\sum_{i}r_{i}=r_{s}}\sum_{i}C(\theta _{i},r_{i}) \\
&=&\min_{\sum_{i}r_{i}=r_{s}} \Big[C(\theta _{n},r_{n}) \\
&&+\min_{\sum_{i<n}r_{i}=r_{s}-r_{n}} \Big[\sum_{i>1}C(\theta
_{n-1},r_{n-1})+\min ... \\
&&+\min_{r_{1}+r_{2}=r_{s}-\sum_{i>2}r_{i}}C(\theta _{2},r_{2})+C(\theta
_{1},r_{1}) \Big]\cdots \Big]
\end{eqnarray*}%
As we showed above, $r_{i}^{\prime } \geq r_{i}^{\ast }$ given $r_{i}^{\prime
}+r_{i-1}^{\prime }=r_{i}^{\ast }+r_{i-1}^{\ast }=r_{i,i-1}$. Thus $C(\theta
_{i},r_{i,i-1})\geq C(\theta _{i},r_{i}^{\prime })+C(\theta
_{i-1},r_{i-1}^{\prime })\geq C(\theta _{i},r_{i}^{\ast })+C(\theta
_{i-1},r_{i-1}^{\ast })$.  By induction, we can prove that $%
\sum_{i}C(\theta _{i},r_{i}^{\prime })\leq C(\max_{i}\theta _{i},r_{s})$.
Now using the same technique as in the proofs of Theorems \ref{Efficiency1} and \ref{Efficiency2},
we obtain the result.
\end{proof}

If either the virtual cost functions $J_i(\theta_i, r_i)$ are not convex, or
the link cost functions and type distributions are not the same across
relays, then higher prices of anarchy may result. Consider the situation in
Figure 2. Here, there are only two relays. $(r_{1}^{\ast },r_{s}-r_{1}^{\ast
})$ is the efficient allocation. Since the type distributions are not the
same, the marginal virtual costs are as indicated in the figure. To minimize
the sum of the virtual costs, the source allocates all traffic to relay 2,
while this allocation is clearly the worst outcome for minimizing the sum of
the link costs.

\begin{figure}[h]
\centering
\includegraphics[scale = 0.6, bb=117 483 345 709]{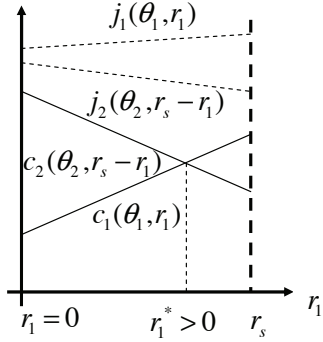}
\caption{Efficiency Loss in Asymmetric Case}
\end{figure}

\section{Conclusion}

This work investigated the impact of incomplete information on incentives
for node cooperation in parallel relay networks. We considered two
situations in which source either has partial bargaining power or full
bargaining power. For the situation where the source has partial bargaining
power, we have shown that all Nash Equilibria in the complete information
game are efficient, including those induced by linear charging functions. We
then characterized the Bayesian Nash Equilibrium for the incomplete
information game in which relays propose linear pricing functions, and
showed that incomplete information can induce inefficiencies, which are
exacerbated by asymmetric prior knowledge on the type distribution. In the
situation where the source has full bargaining power, we first showed that in the
game with complete information, (Bayesian) Nash equilibria exist and are all
efficient. Next, we investigated the game with incomplete information. To
deal with the difficulty of characterizing the Bayesian Nash Equilibria in
this case, we first showed that if a resource allocation outcome can be
realized by a Bayesian Nash equilibrium, then there exists a ``truth
telling" Bayesian Nash equilibrium that realizes the outcome.  We then showed
that the set of outcomes for the ``truth telling" Bayesian Nash equilibria is
included in the set of outcomes for
the Nash equilibria for a {complete information game}, in
which the link cost functions are replaced by a specified ``virtual cost
functions." Using this approach, we obtained for a symmetric network
scenario a bound on the amount of inefficiency which may result from
incomplete information.

\section{Appendix}

{\em Proof of Theorem~\ref{thm:virtual}}:
The first and second-order conditions for~\eqref{eq:truth1} are:
\begin{equation}
\left.\frac{dE_{\theta _{-i}}\left\{ U_{i}(\theta _{i},f(\widetilde{\theta
_{i}},\theta _{-i}))\right\} }{d\widetilde{\theta _{i}}}\right|_{\widetilde{%
\theta _{i}}=\theta _{i}}=0\text{ \ }
\end{equation}
and
\begin{equation}
\left.\frac{d^{2}E_{\theta _{-i}}\left\{ U_{i}(\theta _{i},f(\widetilde{%
\theta _{i}},\theta _{-i}))\right\} }{d\widetilde{\theta _{i}}^{2}}\right|_{%
\widetilde{\theta _{i}}=\theta _{i}}\leq 0.
\end{equation}
The first-order condition is equivalent to
\begin{eqnarray}
&&\left. E_{\theta _{-i}}\frac{dt_{i}(\widetilde{\theta _{i}},\theta _{-i})}{%
d\widetilde{\theta _{i}}}\right\vert _{\widetilde{\theta _{i}}=\theta _{i}}
\\
&=&\left. E_{\theta _{-i}}\left\{ \frac{\partial C_{i}(\theta _{i},r_{i}(%
\widetilde{\theta _{i}},\theta _{-i}))}{\partial r_{i}}\frac{dr_{i}(%
\widetilde{\theta _{i}},\theta _{-i})}{d\widetilde{\theta _{i}}}\right\}
\right\vert _{\widetilde{\theta _{i}}=\theta _{i}}.
\end{eqnarray}
The second-order condition is equivalent to
\begin{eqnarray}
&&\left. E_{\theta _{-i}}\frac{d^{2}t_{i}(\widetilde{\theta _{i}},\theta
_{-i})}{d\widetilde{\theta _{i}}^{2}}\right\vert _{\widetilde{\theta _{i}}%
=\theta _{i}}  \notag \\
&\leq &E_{\theta _{-i}}\left\{ \frac{\partial ^{2}C_{i}(\theta _{i},r_{i}(%
\widetilde{\theta _{i}},\theta _{-i}))}{\partial r_{i}^{2}}\left[ \frac{%
dr_{i}(\widetilde{\theta _{i}},\theta _{-i})}{d\widetilde{\theta _{i}}}%
\right] ^{2}\right. \\
&&+\left. \left. \frac{\partial C_{i}(\theta _{i},r_{i}(\widetilde{\theta
_{i}},\theta _{-i}))}{\partial r_{i}}\frac{d^{2}r_{i}(\widetilde{\theta _{i}}%
,\theta _{-i})}{d\widetilde{\theta _{i}}^{2}}\right\} \right\vert _{%
\widetilde{\theta _{i}}=\theta _{i}}.  \notag
\end{eqnarray}

By evaluating the first-order condition at $\theta _{i}$
differentiating with respect to $\theta_i$, we get:
\begin{eqnarray}
&&E_{\theta _{-i}}\left\{ \frac{d^{2}\left\{ t_{i}(\theta _{i},\theta
_{-i})\right\} }{d\theta _{i}^{2}}\right\}  \notag \\
&=& E_{\theta _{-i}} \left\{\frac{\partial ^{2}C_{i}(\theta _{i},r_{i}(\theta
_{i},\theta _{-i}))}{\partial r_{i}^{2}}\left[ \frac{dr_{i}(\theta
_{i},\theta _{-i})}{d\theta _{i}}\right] ^{2}\right. \\
&&+\frac{\partial C_{i}(\theta _{i},r_{i}(\theta _{i},\theta _{-i}))}{%
\partial r_{i}}\frac{d^{2}r_{i}(\theta _{i},\theta _{-i})}{d\theta _{i}^{2}}
\notag \\
&&+\left. \frac{\partial ^{2}C_{i}(\theta _{i},r_{i}(\theta _{i},\theta
_{-i}))}{\partial r_{i}\partial \theta _{i}}\frac{dr_{i}(\theta _{i},\theta
_{-i})}{d\theta _{i}}\right\}.  \notag
\end{eqnarray}
Comparing with the second-order condition, we get
\begin{equation}
E_{\theta _{-i}}\frac{\partial ^{2}C_{i}(\theta _{i},r_{i}(\theta
_{i},\theta _{-i}))}{\partial r_{i}\partial \theta _{i}}\frac{dr_{i}(\theta
_{i},\theta _{-i})}{d\theta _{i}}\leq 0.
\label{eq:secondcondition}
\end{equation}

We have already assumed that

\begin{equation}
\frac{\partial ^{2}C_{i}(\theta _{i},r_{i}(\theta _{i},\theta _{-i}))}{%
\partial r_{i}\partial \theta _{i}}\leq 0\text{ for each }\theta _{-i}.
\label{eq:assume1}
\end{equation}

Notice that when an outcome can be realized by a Bayesian Nash Equilibrium,
the following condition must hold:
\begin{equation}
\frac{\partial r_{i}(\theta _{i},\theta _{-i})}{\partial \theta _{i}} \geq 0%
\text{ given any }\theta _{-i}
\label{eq:assume2}
\end{equation}
Otherwise, the source would allocate a higher rate to a
lower type relay, which is not optimal.  Notice that by~\eqref{eq:assume1} and~\eqref{eq:assume2},
~\eqref{eq:secondcondition} automatically holds.

Thus, the following conditions are necessary for the first and second-order
conditions to hold.
\begin{eqnarray*}
&&\frac{dE_{\theta _{-i}}\left\{ t_{i}(\widetilde{\theta _{i}},\theta
_{-i})\right\} }{d\widetilde{\theta _{i}}} \\
&=&\left. E_{\theta _{-i}}\frac{\partial C_{i}(\theta _{i},r_{i}(\widetilde{%
\theta _{i}},\theta _{-i}))}{\partial r_{i}}\frac{dr_{i}(\widetilde{\theta
_{i}},\theta _{-i})}{d\widetilde{\theta _{i}}}\right\vert _{\widetilde{%
\theta _{i}}=\theta _{i}}
\end{eqnarray*}

\begin{equation*}
\frac{\partial r_{i}(\theta _{i},\theta _{-i})}{\partial \theta _{i}}\geq 0%
\text{ given any }\theta _{-i}
\end{equation*}

Let $V_{i}(\theta _{i},\theta _{-i})=\max_{{\widetilde{\theta _{i}}}%
}U_{i}(\theta _{i},r_{i}({\widetilde{\theta _{i}},\theta _{-i}}),t_{i}({%
\widetilde{\theta _{i}},\theta _{-i}}))$.   We use the envelope theorem
just as we did in the previous sections:
\begin{eqnarray}
\frac{dE_{\theta _{-i}}V_{i}(\theta _{i},\theta _{-i})}{d\theta _{i}}
&=&\left.\frac{\partial {E_{\theta _{-i}}U_{i}(\theta _{i},r}_{i}{(%
\widetilde{\theta _{i}},\theta _{-i}),t}_{i}{(\widetilde{\theta _{i}},\theta
_{-i})}}{\partial {\theta _{i}}}\right|_{\widetilde{\theta _{i}}=\theta _{i}}
\notag \\
&=&\left.-\frac{\partial {E_{\theta _{-i}}C_{i}(\theta _{i},r_{i}(\widetilde{%
\theta _{i}},\theta _{-i}))}}{\partial \theta _{i}}\right|_{\widetilde{%
\theta _{i}}=\theta _{i}}
\end{eqnarray}


Let $\overline{\theta_i}$ and $\underline{\theta_i}$ be the upper and
lower bounds on relay node i's type, then
\begin{eqnarray}
&&E_{\theta _{-i}}V_{i}(\theta _{i},\theta _{-i}) \\
&=&E_{\theta _{-i}}V_{i}(\underline{\theta _{i}},\theta _{-i})-\int_{%
\underline{\theta _{i}}}^{\theta _{i}}\frac{\partial E_{\theta _{-i}}{%
C_{i}(\theta _{i},r_{i}(\theta _{i},\theta _{-i}))}}{\partial {\theta _{i}}}%
d\theta _{i}  \notag
\end{eqnarray}

We see from the above equation that, as we already assumed $\frac{%
\partial {C_{i}(\theta _{i},r_{i})}}{\partial {\theta _{i}}}<0$, the
expected utility of relay $i$ is non-decreasing with respect to $\theta _{i}$.
Thus, to guarantee that constraints (\ref{eq:nonnegative}) holds, the lowest
type must receive non-negative profit.  On the other hand, the relay with the
lowest type can never receive a positive profit, otherwise the source will
reduce its profit by some small amount and still guarantee that the contract is
acceptable to all, which contradicts the definition of Bayesian Nash Equilibrium.
Thus, the lowest type relay should receive zero profit.
\begin{equation}
E_{\theta _{-i}}V_{i}(\underline{\theta _{i}},\theta _{-i})=0
\end{equation}%
Plugging in, we get
\begin{equation}
E_{\theta _{-i}}V_{i}(\theta _{i},\theta _{-i})=-\int_{\underline{\theta _{i}%
}}^{\theta _{i}}\frac{\partial E_{\theta _{-i}}{C_{i}(\theta
_{i},r_{i}(\theta _{i},\theta _{-i}))}}{\partial {\theta _{i}}}d\theta _{i}
\end{equation}

Suppose the type distribution function of relay $i$ is $F_{i}(\theta _{i})$ and the
density is $f_{i}(\theta _{i})$. Let $R$ be the expected revenue of the source
node.  Then,

\begin{eqnarray*}
R &=&E_{\theta }\left\{ W_{s}(r_{s})-C_{s}(\theta
_{s},r_{0})-\sum_{i}t_{i}(\theta )\right\}  \\
&=&E_{\theta }\left\{ W_{s}(r_{s})-C_{s}(\theta
_{s},r_{0})-\sum_{i}V_{i}(\theta )-\sum_{i}C_{i}(\theta )\right\}  \\
&=&E_{\theta }\left\{ W_{s}(r_{s})-C_{s}(\theta
_{s},r_{0})-\sum_{i}C_{i}(\theta )\right\}  \\
&&+\sum_{i}\int_{\underline{\theta _{i}}}^{\overline{\theta _{i}}%
}f_{i}(\theta _{i})E_{\theta _{-i}}\left[ \int_{\underline{\theta _{i}}%
}^{\theta _{i}}\frac{\partial C_{i}(\theta _{i}^{\prime },r_{i}(\theta
_{i}^{\prime },\theta _{-i}))}{\partial \theta _{i}^{\prime }}d\theta
_{i}^{\prime }\right] d\theta _{i} \\
&=&E_{\theta }\left\{ W_{s}(r_{s})-C_{s}(\theta
_{s},r_{0})-\sum_{i}C_{i}(\theta )\right\}  \\
&&-\sum_{i}\int_{\underline{\theta _{i}}}^{\overline{\theta _{i}}}E_{\theta
_{-i}}\left[ \int_{\underline{\theta _{i}}}^{\theta _{i}}\frac{\partial
C_{i}(\theta _{i}^{\prime },r_{i}(\theta _{i}^{\prime },\theta _{-i}))}{%
\partial \theta _{i}^{\prime }}d\theta _{i}^{\prime }\right]  \\
&&\times d(1-F_{i}(\theta _{i})) \\
&=&E_{\theta }\left\{ W_{s}(r_{s})-C_{s}(\theta
_{s},r_{s})-\sum_{i}C_{i}(\theta )\right\}  \\
&&-\sum_{i}E_{\theta _{-i}}\left[ \int_{\underline{\theta _{i}}}^{\theta
_{i}}\frac{\partial C_{i}(\theta _{i}^{\prime },r_{i}(\theta _{i}^{\prime
},\theta _{-i}))}{\partial \theta _{i}^{\prime }}d\theta _{i}^{\prime }%
\right]  \\
&&\times (1-F_{i}(\theta _{i}))|_{\underline{\theta _{i}}}^{\overline{\theta
_{i}}} \\
&&+\sum_{i}E_{\theta _{-i}}\int_{\underline{\theta _{i}}}^{\overline{\theta
_{i}}}(1-F_{i}(\theta _{i})) \\
&&\times d\left[ \int_{\underline{\theta _{i}}}^{\theta _{i}}\frac{\partial
C_{i}(\theta _{i}^{\prime },r_{i}(\theta _{i}^{\prime },\theta _{-i}))}{%
\partial \theta _{i}^{\prime }}d\theta _{i}^{\prime }\right]  \\
&=&E_{\theta }\left\{ W_{s}(r_{s})-C_{s}(\theta
_{s},r_{0})-\sum_{i}C_{i}(\theta )\right\}  \\
&&+\sum_{i}E_{\theta _{-i}}\int_{\underline{\theta _{i}}}^{\overline{\theta
_{i}}}(1-F_{i}(\theta _{i})) \\
&&\times d\left[ \int_{\underline{\theta _{i}}}^{\theta _{i}}\frac{\partial
C_{i}(\theta _{i}^{\prime },r_{i}(\theta _{i}^{\prime },\theta _{-i}))}{%
\partial \theta _{i}^{\prime }}d\theta _{i}^{\prime }\right]  \\
&=&E_{\theta }\{W_{s}(r_{s})-C_{s}(\theta _{s},r_{0})\} \\
&&-\sum_{i}E_{\theta _{-i}}\int_{\underline{\theta _{i}}}^{\overline{\theta
_{i}}}C_{i}(\theta _{i},r_{i}(\theta _{i},\theta _{-i})) \\
&&-\frac{1-F_{i}(\theta _{i})}{f_{i}(\theta _{i})}\frac{\partial
C_{i}(\theta _{i},r_{i}(\theta ))}{\partial \theta _{i}}f_{i}(\theta
_{i})d\theta _{i} \\
&=&E_{\theta }[W_{s}(r_{s})-C_{s}(\theta _{s},r_{0})] \\
&&-E_{\theta }\sum_{i}\left( C_{i}(\theta _{i},r_{i}(\theta ))-\frac{%
1-F_{i}(\theta _{i})}{f_{i}(\theta _{i})}\frac{\partial C_{i}(\theta
_{i},r_{i}(\theta ))}{\partial \theta _{i}}\right)
\end{eqnarray*}

Thus, we obtain a game with complete information and full source
bargaining power where the revenue function is changed to
$J_{i}(\theta _{i},r_{i})$ rather than $C_{i}(\theta _{i},r_{i})$.

\bibliographystyle{ieeetr}
\bibliography{./gameNet,./ICCBiblio}

\end{document}